\documentclass{scrartcl}

\usepackage{amsmath,amssymb}

\usepackage{amsmath,amsthm,amssymb,amsfonts,color}
\usepackage[latin1]{inputenc}

\newtheorem{theorem}{Theorem}[section]\newtheorem{lemma}[theorem]{Lemma}

\newtheorem{proposition}[theorem]{Proposition}
\newtheorem{corollary}[theorem]{Corollary}
\newtheorem{remark}[theorem]{Remark}

\newcommand{\norm}[1]{\left\| #1 \right\|}  %Norm
\newcommand{\N}{\mathbb{N}}  %Zahlbereiche

\newcommand{\R}{\mathbb{R}}

\newcommand{\eps}{\varepsilon}

\newcommand{\ol}{\overline}

\DeclareFontFamily{U}{mathx}{\hyphenchar\font45}
\DeclareFontShape{U}{mathx}{m}{n}{
      <5> <6> <7> <8> <9> <10>
      <10.95> <12> <14.4> <17.28> <20.74> <24.88>
      mathx10
      }{}
\DeclareSymbolFont{mathx}{U}{mathx}{m}{n}
\DeclareFontSubstitution{U}{mathx}{m}{n}
\DeclareMathAccent{\widecheck}{0}{mathx}{"71}
\DeclareMathAccent{\wideparen}{0}{mathx}{"75}

\title{On the dynamics of the mean-field polaron in the weak-coupling limit}

\author{Marcel Griesemer, Jochen Schmid, Guido Schneider\\  
\small Fachbereich Mathematik, Universit\"at Stuttgart, D-70569 Stuttgart, Germany\\
\small firstname.lastname@mathematik.uni-stuttgart.de}

\date{}

\begin{document}

\maketitle

\begin{abstract}
We consider the dynamics of the mean-field polaron in the weak-coupling limit of vanishing electron-phonon interaction, $ \varepsilon \to 0$. This is a singular limit formally leading to a Schr\"odinger--Poisson system that is equivalent to the nonlinear Choquard equation. 
By establishing estimates between the  approximation obtained via the Choquard equation and true solutions of the original system we show  that the Choquard equation makes correct predictions about the dynamics of the polaron mean-field model for small values of $\varepsilon > 0 $. 
\end{abstract}

\section{Introduction}

In Pekar's model of (large) polarons a single electron interacts with a dielectric polarizable elastic medium. The polarization of the medium by the charge of the electron creates an electrostatic potential that, in turn, acts on the electron. In the stationary case this leads to a self-trapped state of the electron called polaron \cite{DA09, Lieb1977}. In the non-stationary case the electron triggers harmonic oscillation of the elastic medium and the combined system is described by the coupled equations
\begin{align} \label{polaron1}
i \partial_t u &=  -\Delta u + vu ,\\
\varepsilon^2 \partial^2_t v &= - v + \Delta^{-1} |u|^2, \label{polaron2}
\end{align}
for the wave function $ u = u(x,t) \in \mathbb{C} $ of  the electron and the electrostatic potential  $ v = v(x,t) \in \mathbb{R} $ associated with the  polarization of the medium. Here
$ x \in \mathbb{R}^3 $, $ t \in \mathbb{R} $, and $\Delta^{-1}|u|^2 = -(4\pi|\cdot|)^{-1} * |u|^2$. The parameter $\varepsilon > 0 $ plays the role of  the electron-phonon coupling strength or the inverse of the phonon frequency in a more fundamental, quantum field theoretic model of the polaron \cite{FG2015, BNRS00}.

The presence of the term $\varepsilon^2 \partial^2_t v$ in \eqref{polaron2} leads to retardation in the self-interaction, which makes it difficult to predict 
 the evolution of the electron. We are therefore interested in the question whether \eqref{polaron1}-\eqref{polaron2} may be solved approximately by dropping $\varepsilon^2 \partial^2_t v$ if  $\varepsilon$ is small. This approach leads to the Schr\"odinger--Poisson system
\begin{align} 
i \partial_t U &= -\Delta U + VU, \label{eq: hs1}\\
0 &= - V + \Delta^{-1} |U|^2, \label{eq: hs2}
\end{align}
which is equivalent to the Choquard equation
\begin{equation} \label{hartree}
i \partial_t U = -\Delta U + (\Delta^{-1} |U|^2) U. 
\end{equation}
Such nonlinear Hartree-type equations admit an interpretation as infinite-dimensional Hamiltonian systems, and stationary points of the associated Hamilton functionals lead to solitary-wave solutions \cite{FTY2002, LiLin2010}. In the present case, assuming the legitimacy of letting $\eps\to 0$, solitary-wave solutions describe frictionlessly moving polarons. The Choquard equation \eqref{hartree} in a different context describes the evolution of coherent states (condensates) of bosons in the mean-field limit, and it has been proposed as a model for gravity-induced decoherence \cite{EY2001, Penrose1998}.  Similar Hartree-type nonlinear equations arise in many further areas of mathematical physics.

The goal of the present paper is to prove the following approximation theorem.

\begin{theorem}\label{maintheorem}
If $U \in C([0,T_0],H^4(\R^3))$ is a solution of~\eqref{hartree} and $V = \Delta^{-1}|U|^2$, then for $C_1 >0$ there exist $C_2 > 0$ and $\eps_0' > 0$ such that
for all $\eps \in (0,\eps_0']$ and all solutions $u_{\eps}, v_{\eps}$ of~\eqref{polaron1} and \eqref{polaron2} satisfying
\begin{align*}
\norm{u_{\eps}(\cdot,0)-U(\cdot,0)}_{H^2} + \norm{v_{\eps}(\cdot,0)-V(\cdot,0)}_{L^{\infty}} + \norm{ \Delta \big( v_{\eps}(\cdot,0)-V(\cdot,0) \big) }_{L^2 \cap L^1} 
\le C_1 \eps
\end{align*}
and $$\norm{\partial_t v_{\eps}(\cdot,0)}_{L^{\infty}} + \norm{  \Delta \big( \partial_t v_{\eps}(\cdot,0) \big) }_{L^2 \cap L^1} \le C_1,$$ we have
\begin{align*}
\norm{u_{\eps}(\cdot,t)-U(\cdot,t)}_{H^2} + \norm{v_{\eps}(\cdot,t)-V(\cdot,t)}_{L^{\infty}}  + \norm{ \Delta \big( v_{\eps}(\cdot,t)-V(\cdot,t) \big) }_{L^2 \cap L^1} 
\le C_2 \eps
\end{align*}
for all $t \in [0,T_0]$. 
\end{theorem}

See Section~\ref{sec: approx result} for  a more general version of this result -- and a more precise formulation specifying, for instance, the notion of solution employed above and the spaces the solutions live in. 
\begin{remark}
    {\rm Such approximation results should not be taken for granted. There are various counterexamples showing 
    that formally derived limit equations make wrong predictions about the original system \cite{Schn95MN,SSZ14}.}
\end{remark}
\begin{remark}\label{remark2}
     {\rm The approximation result is non-trivial since the Lipschitz constant of the right-hand side of the  first-order system, cf. \eqref{eq: int eq R_v,R_w}, 
associated to \eqref{polaron2} is of order $ \mathcal{O}(\varepsilon^{-1})$.
Hence, especially the nonlinear terms in \eqref{polaron2} in principle can lead to some unwanted growth rates $ \mathcal{O}(e^{\varepsilon^{-1} t}) $ for  $ t = \mathcal{O}(1) $.}
\end{remark}

The problem described in Remark~\ref{remark2} is overcome by an integration by parts w.r.t.~$ t $ in the variation of constants formula
associated to  \eqref{polaron2} and by a well adapted choice of spaces and norms. This allows us to use 
the highly oscillatory linear semigroup associated to  \eqref{polaron2} to get rid of the $ \varepsilon^{-1} $
in front of the nonlinear terms, cf. also Section \ref{remnormalform}. Our estimates imply, in particular, the existence of an interval $[0,T_0]$, that is independent of $\eps \in (0,\eps_0']$, on which the system has a unique (mild) solution.

While it appears natural, mathematically, to study the limit $\varepsilon\to 0$ of \eqref{polaron1}-\eqref{polaron2}, from a physical point of view the limit $\varepsilon \to\infty$ is even more relevant, because  the system \eqref{polaron1}-\eqref{polaron2}
is believed, and partly proven, to describe the strong coupling limit, $\varepsilon \to\infty$, of the Fr\"ohlich model of large polarons \cite{FS14, FG2015}. We remark that the Nelson model, which is similar to the Fr\"ohlich model, in a classical limit leads to the  Schr\"odinger-Klein-Gordon system \cite{AF2014}. The designation of \eqref{polaron1}-\eqref{polaron2} as \emph{mean-field polaron} in the title of our paper is adopted from  \cite{BNRS00}, where these equation, apparently, have been studied for the first time.

We conclude this introduction with some remarks on our notation. 
The intersection $X \cap Y$ and the product $X \times Y$ of two function spaces $X, Y$ on $\R^3$ will always be endowed with the sum norm $\norm{\cdot}_{X \cap Y} := \norm{\cdot}_X + \norm{\cdot}_Y$ and the product norm $\norm{(\cdot,\cdot \cdot)}_{X \times Y} := \norm{\cdot}_X + \norm{\cdot \cdot}_Y$, respectively. The operator norm for bounded operators between $X$ and $Y$ will be denoted by $\norm{\cdot}_{X,Y}$. We use $C_0(\R^3)$ to denote the space of continuous functions tending to $0$ at infinity, while compact support will be indicated by the notation $C_c^k(\R^3)$ for $k \in \N_0 \cup \{\infty\}$. Finally, distinct constants will be denoted with the same letter $C$ if they can be chosen independently of the small perturbation parameter $\varepsilon \ll 1 $.
\medskip

{\bf Acknowledgement.} The work of J.~Schmid  is partially supported by the Deutsche Forschungsgemeinschaft DFG 
through the Graduiertenkolleg GRK 1838 ,,Spectral Theory and Dynamics of Quantum Systems''.

\section{Spaces and operators} \label{sec: some preliminaries}

In this section, we introduce the spaces we will work with and investigate the mapping properties of the Laplace operator $\Delta$ in theses spaces. In particular, we will discuss the properties of the inverse operator $\Delta^{-1}$ appearing in the equations~\eqref{polaron2} and~\eqref{hartree}.
If $X, Y \subset L^1_{\text{loc}}(\R^3)$ are function spaces on $\R^3$, we will write $\Delta: D_{X,Y} \subset X \to Y$ to denote the linear operator %for the operator 
$D_{X,Y} \ni u \mapsto \Delta u \in Y$
with domain
\begin{align*}
D_{X,Y} := \{ u \in X: \Delta u \in Y \},
\end{align*}
where $\Delta u = \partial_{x_1}^2 u+ \partial_{x_2}^2 u+ \partial_{x_3}^2 u$ denotes the distributional Laplacian of $u$. 
We will continually use the Sobolev spaces 
\begin{align*}
     X_s := H^s(\R^3) 
\quad \text{with} \quad \norm{u}_{X_s} := \Big( \int (1+|\xi|^2)^s \, |\widehat{u}(\xi)|^2 \,d\xi \Big)^{1/2}
\end{align*}
for $s \in [0,\infty)$, whose basic and completely well-known properties are summarized in the following lemma for the sake of easy reference.

\begin{lemma} \label{lm: X_s basics}
For $s,t \in [0,\infty)$, the following holds true:
\begin{itemize}
\item[(i)] $X_{s}$ is a Hilbert space which is continuously embedded in $X_{t}$ for all $t \le s$.
\item[(ii)] If $s \in [2,\infty)$, then $X_s \cdot X_s \subset X_s \cap L^1(\R^3)$ and there is a constant $C = C_s$ such that $\norm{u v}_{X_s \cap L^1} \le C \norm{u}_{X_s} \norm{v}_{X_s}$ for all $u, v \in X_s$.
\item[(iii)] $D_{X_{s},X_{s}} = X_{s+2}$ and $\Delta: X_{s+2} \subset X_{s} \to X_{s}$ is a self-adjoint linear operator satisfying $\norm{\Delta u}_{X_{s}} \le  \norm{u}_{X_{s+2}}$ for all $u \in X_{s+2}$. 
\end{itemize}
\end{lemma}

We will also need the following refinement of Lemma \ref{lm: X_s basics} (ii).

\begin{lemma} \label{lm: X_s-2 X_s in X_s-2}
If $s \in [2,\infty)$, then $X_{s-2} \cdot X_s \subset X_{s-2} \cap L^1(\R^3)$ and there is a constant $C = C_s$ such that
\begin{align*}
\norm{u v}_{X_{s-2} \cap L^1} \le C \norm{u}_{X_{s-2}} \norm{v}_{X_s} \qquad (u \in X_{s-2}, v \in X_s).
\end{align*}
\end{lemma}

\begin{proof}
We have only to show the inclusion $X_{s-2} \cdot X_s \subset X_{s-2}$ and the estimate for the $X_{s-2}$-norm because the respective inclusion and estimate for $L^1(\R^3)$ are an immediate consequence of Schwarz's inequality. It follows from Theorem~9.3.5 in~\cite{Friedlander} that for $u \in X_{s-2}$ and $v \in \mathcal{S}(\R^3)$ the product $uv$ belongs to $X_{s-2}$ with
\begin{align*}
\norm{uv}_{X_{s-2}} &\le C \norm{u}_{X_{s-2}} \int (1+|\xi|^2)^{(s-2)/2} |\widehat{v}(\xi)| \,d\xi \\
&\le C \Big( \int (1+|\xi|^2)^{-2} \,d\xi \Big)^{1/2} \norm{u}_{X_{s-2}} \norm{v}_{X_s} 
\le C \norm{u}_{X_{s-2}} \norm{v}_{X_s}.
\end{align*}
An obvious approximation argument now yields the assertion. 
\end{proof}

In the next lemma, we deal with the invertibility of $\Delta$ and the elementary properties of the inverse. 

\begin{lemma} \label{lm: inverse of laplace}
$\Delta: D_{C_0,L^2 \cap L^1} \subset C_0(\R^3) \to L^2(\R^3) \cap L^1(\R^3)$ is an invertible linear operator with full range
and bounded inverse $\Delta^{-1}$ satisfying
\begin{align} \label{eq: inverse of laplace}
\norm{\Delta^{-1}(w)}_{C_0} \le C \big( \norm{w}_{L^2} + \norm{w}_{L^1} \big) \quad \text{and} \quad \Delta^{-1}(w) = \gamma * w
\end{align}
for all $w \in L^2(\R^3) \cap L^1(\R^3)$, where $\gamma$ is the fundamental solution of Laplace's equation in~$\R^3$ with $\gamma(x) = -1/(4 \pi |x|)$ for $x \in \R^3 \setminus \{0\}$, and where $C$ is a constant independent of $w$.
\end{lemma}

\begin{proof}
Injectivity is a simple exercise using the structure theorem for distributions with support contained in $\{0\}$. Surjectivity and the properties of $\Delta^{-1}$ are equally simple. Indeed, if $w \in L^2(\R^3) \cap L^1(\R^3)$, then 
\begin{align} \label{eq: 1, inverse of laplace}
\int \frac{|\widehat{w}(\xi)|}{|\xi|^2} \,d\xi &\le \Big( \int_{|\xi| > 1} \frac{1}{|\xi|^4} \,d\xi \Big)^{1/2} \norm{\widehat{w}}_{L^2} + \Big( \int_{|\xi| \le 1} \frac{1}{|\xi|^2} \,d\xi \Big) \norm{\widehat{w}}_{C_0} \notag \\
&\le C ( \norm{w}_{L^2} + \norm{w}_{L^1} )
\end{align}
by the continuity of the Fourier transform from $L^1(\R^3)$ to $C_0(\R^3)$. So, $\widehat{w}/|\cdot|^{2} \in L^1(\R^3)$ and thus $v := - \,\,\widecheck\,\,(\widehat{w}/|\cdot|^{2}) \in C_0(\R^3)$ and, of course, $\Delta v = w$,
which proves that $\Delta: D_{C_0,L^2 \cap L^1} \subset C_0(\R^3) \to L^2(\R^3) \cap L^1(\R^3)$ is surjective and that
\begin{align} \label{eq: 2, inverse of laplace}
\Delta^{-1}(w) = - \,\,\widecheck\,\,\big(\widehat{w}/|\cdot|^{2}\big) \qquad (w \in L^2(\R^3) \cap L^1(\R^3)).
\end{align}
Combining \eqref{eq: 1, inverse of laplace} and \eqref{eq: 2, inverse of laplace}, we obtain the estimate in~\eqref{eq: inverse of laplace}. Additionally, we obtain from~\eqref{eq: 2, inverse of laplace} the convolution 
representation of $\Delta^{-1}(w)$ in~\eqref{eq: inverse of laplace} by virtue of the convolution theorem for tempered distributions (together with a suitable approximation argument).
\end{proof}

In the following, $\Delta^{-1}$ will always denote the operator from the lemma above or a restriction of that operator.  In order to control the nonlinear terms 
in \eqref{polaron2} and \eqref{hartree} we use:

\begin{lemma} \label{lm: u times delta^(-1)(w)}
If $s \in [2,\infty)$, then for all $u \in X_s$ and $w \in X_{s-2} \cap L^1(\R^3)$ one has $u \, \Delta^{-1}(w) \in X_s$ and 
\begin{align} \label{eq: u times ...}
\norm{u \, \Delta^{-1}(w)}_{X_s} \le C \norm{u}_{X_s} \big( \norm{w}_{X_{s-2}} + \norm{w}_{L^1} \big),
\end{align}
where $C=C_s$ is a constant independent of $u$ and $w$.
\end{lemma}

\begin{proof}
Clearly, we have to show the assertion only for $u \in \mathcal{S}(\R^3)$. So let $u \in \mathcal{S}(\R^3)$ and $w \in X_{s-2} \cap L^1(\R^3)$ and set $v := \Delta^{-1}(w)$. Also, write 
\begin{align*}
\rho_{s}(\xi) := (1+|\xi|^2)^{s/2} \quad \text{and} \quad \sigma_{s}(\xi) := |\xi|^s
\end{align*}
for $\xi \in \R^3$. Since $\widehat{v} = -\, \widehat{w}/|\cdot|^{2}$ by~\eqref{eq: 2, inverse of laplace}, it follows from~\eqref{eq: 1, inverse of laplace} that $\widehat{v} \in L^1(\R^3)$ and $\sigma_{s} \, \widehat{v} \in L^2(\R^3)$ with
\begin{align} \label{eq: u times ..., 1}
\norm{\widehat{v}}_{L^1} \le C ( \norm{w}_{X_{s-2}} + \norm{w}_{L^1} )
\quad  \text{and} \quad
\norm{\sigma_{s} \, \widehat{v}}_{L^2} \le \norm{w}_{X_{s-2}}
\end{align}
where $C$ is a constant independent of $w$ (and $s$). So, $\widehat{u} \in \mathcal{S}(\R^3)$ and $\widehat{v} \in L^1(\R^3) \subset \mathcal{S}'(\R^3)$ are classically convolvable and thus, by the convolution theorem for tempered distributions, we see that
\begin{align}
%\widehat \,\,(uv)(x) = (2\pi)^{-3/2}\, (\widehat{u} * \widehat{v})(x) = (2\pi)^{-3/2} \int \widehat{u}(x-y) \widehat{v}(y) \,dy \qquad (x \in \R^3).
\widehat{uv}(x) = (2\pi)^{-3/2}\, (\widehat{u} * \widehat{v})(x) = (2\pi)^{-3/2} \int \widehat{u}(x-y) \widehat{v}(y) \,dy \qquad (x \in \R^3).
\end{align}
Since $\rho_{s}(x) \le C \rho_{s}(x-y) + C \sigma_{s}(y)$ for all $x,y \in \R^3$ with $C = 2^s$, it follows that
\begin{align*}
%\rho_{s}(x) |\, \widehat \,\,(uv)(x)| \le C \big( (\rho_{s} |\widehat{u}|) * |\widehat{v}| \big)(x) + C \big( |\widehat{u}| * (\sigma_{s} |\widehat{v}|) \big)(x) 
\rho_{s}(x) | \widehat{uv}(x)| \le C \big( (\rho_{s} |\widehat{u}|) * |\widehat{v}| \big)(x) + C \big( |\widehat{u}| * (\sigma_{s} |\widehat{v}|) \big)(x)
\end{align*}
for all $x \in \R^3$ and therefore
\begin{align}
%\Big( \int \big( \rho_{s}(x) |\, \widehat \,\,(uv)(x)| \big)^2 \,dx \Big)^{1/2} \le C \norm{ \rho_{s} |\widehat{u}|}_{L^2} \norm{\widehat{v}}_{L^1} + C \norm{\widehat{u}}_{L^1} \norm{\sigma_{s} |\widehat{v}|}_{L^2} \notag \\
%\le C \norm{u}_{X_s} \big( \norm{w}_{X_{s-2}} + \norm{w}_{L^1} \big) + C \norm{u}_{X_s} \big( \norm{w}_{X_{s-2}} + \norm{w}_{L^1} \big)
\Big( \int \big( \rho_{s}(x) | \widehat{uv}(x)| \big)^2 \,dx \Big)^{1/2} \le C \norm{ \rho_{s} |\widehat{u}|}_{L^2} \norm{\widehat{v}}_{L^1} + C \norm{\widehat{u}}_{L^1} \norm{\sigma_{s} |\widehat{v}|}_{L^2} \notag \\
\le C \norm{u}_{X_s} \big( \norm{w}_{X_{s-2}} + \norm{w}_{L^1} \big) + C \norm{u}_{X_s} \norm{w}_{X_{s-2}} 
\end{align}
by Young's inequality and by the inequalities~\eqref{eq: u times ..., 1}. So, we have $u \, \Delta^{-1}(w) \in X_s$ and the estimate~\eqref{eq: u times ...} holds true, as desired. 
\end{proof}

In view of the above lemmas, we introduce the spaces
\begin{align*}
Y_s := \Delta^{-1}(X_s \cap L^1(\R^3)) \quad \text{with} \quad \norm{v}_{Y_s} := \norm{v}_{C_0} + \norm{\Delta v}_{X_s} + \norm{\Delta v}_{L^1}
\end{align*}
for $s \in [0,\infty)$, whose basic properties are summarized in the following lemma.

\begin{lemma} \label{lm: Y_s basics}
For $s,t \in [0,\infty)$, the following holds true:
\begin{itemize}
\item[(i)] $Y_{s}$ is a Banach space which is continuously embedded in $Y_{t}$ for all $t \le s$.
\item[(ii)] If $s \in [2,\infty)$, then $X_s \cdot Y_{s-2} \subset X_s$ and there is a constant $C = C_s$ such that $\norm{u v}_{X_s} \le C \norm{u}_{X_s} \norm{v}_{Y_{s-2}}$ for all $u \in X_s$ and $v \in Y_{s-2}$.
\item[(iii)] $\Delta^{-1}: X_{s} \cap L^1(\R^3) \to Y_{s}$ is a bounded linear operator.
\end{itemize}
\end{lemma}

\begin{proof}
%We have only to prove assertion~(i) since assertions~(ii) and (iii) are immediate consequences of Lemma~\ref{lm: u times delta^(-1)(w)} and Lemma~\ref{lm: inverse of laplace} respectively. 
Assertion~(i) easily follows by the completeness of $C_0(\R^3)$ and $X_{s} \cap L^1(\R^3)$ and by the boundedness of $\Delta^{-1}: X_{s} \cap L^1(\R^3) \to C_0(\R^3)$ (Lemma~\ref{lm: inverse of laplace}). Assertions~(ii) and (iii) are immediate consequences of Lemma~\ref{lm: u times delta^(-1)(w)} and Lemma~\ref{lm: inverse of laplace} respectively.
\end{proof}

%---------------------------------------------------------------------------------------------------------------------------------------------------------------------------------------------------

\section{Solvability of the equations}

In this section, we discuss the solvability of the equations~\eqref{polaron1}-\eqref{polaron2} and of~\eqref{hartree}, which is of course the very first thing to do in proving the desired approximation result.
We start with  the approximation equation~\eqref{hartree} and first show mild and classical solvability of the corresponding abstract initial value problem
\begin{align} \label{eq: approx eq, abstract}
U' = i\Delta U -i U \Delta^{-1}(|U|^2), \quad \text{with} \quad U(0) = U_0
\end{align}
in the sense of \cite{Pazy}.

\begin{theorem} \label{thm: approx eq mild sol}
If $s \in [0,\infty)$, then for every $U_0 \in X_{s+2}$ there exists a $T_0 > 0$ %a compact interval $I = [0,T_0]$ of positive length
and a unique mild solution $U \in C(I,X_{s+2})$ of~\eqref{eq: approx eq, abstract} on $I = [0,T_0]$. %where $I = [0,T_0]$. 
\end{theorem}

\begin{proof}
Clearly, the linear part $i\Delta$ of the equation~\eqref{eq: approx eq, abstract} is the generator of a strongly continuous unitary group in $X_{s+2}$ by virtue of Lemma~\ref{lm: X_s basics}~(iii). 
Also, the nonlinear part $f$ of the equation~\eqref{eq: approx eq, abstract} given by
\begin{align*}
f(U) := -i U \Delta^{-1}(|U|^2) %\qquad (U \in X_{s+2})
\end{align*}
is a map from $X_{s+2}$ into itself and Lipschitz continuous on bounded subsets by virtue of Lemma~\ref{lm: Y_s basics} (ii) and~(iii).
So, the standard existence and uniqueness result for mild solutions (Theorem~6.1.4 in~\cite{Pazy}) %or Theorem~8.6 in~\cite{ISEM}) 
implies that there is a $T_0 > 0$ and a unique mild solution $U: I = [0,T_0] \to X_{s+2}$ of~\eqref{eq: approx eq, abstract}. In other words, the integral equation
\begin{align} \label{eq: def mild sol, approx eq}
U(t) = e^{i \Delta t}U_0 -i \int_0^t e^{i \Delta (t-r)} U(r) \Delta^{-1}(|U(r)|^2) \,dr 
\end{align}
has a unique solution $U \in C(I,X_{s+2})$. 
\end{proof}

In the situation of the above theorem, we also obtain classical solvability by Theorem~6.1.5 in \cite{Pazy}:
\begin{corollary} \label{cor: approx eq class sol}
If $s \in [0,\infty)$ and if $U_0 \in X_{s+2}$ and $U \in C(I,X_{s+2})$ are as in the above theorem, then $U$ belongs to $C^1(I,X_s)$ and is a classical solution of~\eqref{eq: approx eq, abstract}.
\end{corollary}

We now go on with the original equations~\eqref{polaron1} and~\eqref{polaron2} and 
show mild and classical solvability of the corresponding abstract initial value problem
\begin{align} \label{eq: original eq, abstract}
\begin{pmatrix}u\\v\\w\end{pmatrix}' =  \begin{pmatrix}i\Delta u\\ \eps^{-1}w\\ -\eps^{-1} v\end{pmatrix} + \begin{pmatrix}-i uv\\ 0\\ \eps^{-1} \Delta^{-1}(|u|^2)\end{pmatrix} ,
\quad \text{with} \quad
\begin{pmatrix}u\\v\\w\end{pmatrix}(0) = \begin{pmatrix}u_0\\v_0\\w_0\end{pmatrix} 
\end{align}
in the sense of \cite{Pazy}. 
In the following, we will always abbreviate 
\begin{align*}
\Lambda(v,w) := (w,-v) \quad \text{for} \quad (v,w) \in Y_s\times Y_s. 
\end{align*}

\begin{lemma} \label{lm: group uniformly bd in eps}
If $s\in [0,\infty)$, then $\eps^{-1} \Lambda$ is the generator of a continuous group $(e^{\eps^{-1} \Lambda t})_{t \in \R}$ in $Y_s \times Y_s$, which is uniformly bounded w.r.t.~$\eps \in (0,\infty)$, that is,
\begin{align*}
\sup_{\eps \in (0,\infty), t \in \R} \norm{ e^{\eps^{-1} \Lambda t} }_{Y_s \times Y_s, Y_s \times Y_s} < \infty.
\end{align*}
%$e^{eps^{-1} \Lambda t}$ is a contraction in $Y_s \times Y_s$ for every $s \in [0,\infty), t \in \R$ and $\eps >0$. Hierzu m\"usste Y_s \times Y_s mit der l^2(Y_s \times Y_s)-Norm versehen sein 
\end{lemma}

\begin{proof}
Since $\Lambda$ %defined in~\eqref{eq: lin part, original eq} 
is a bounded operator in $Y_s \times Y_s$ %with matrix representation $\Lambda = \begin{pmatrix}0 & 1 \\ -1 & 0 \end{pmatrix}$,
having the matrix representation
\begin{align*}
\Lambda = \begin{pmatrix}0 & 1 \\ -1 & 0 \end{pmatrix},
\end{align*}
we have the explicit representation formula
\begin{align*}
e^{\eps^{-1} \Lambda t} = \sum_{n=0}^{\infty} \frac{\big( \eps^{-1} \Lambda t \big)^n}{n!} = \begin{pmatrix}
\cos(\eps^{-1}t)  & \sin(\eps^{-1}t) \\
-\sin(\eps^{-1}t) & \cos(\eps^{-1}t)
\end{pmatrix}
\end{align*}
from which the assertion is obvious. 
%and from this the assertion is obvious. 
\end{proof}

\begin{theorem} \label{thm: original eq mild sol}
If $s \in [2,\infty)$ and $\eps > 0$, then for every $(u_0,v_0,w_0) = (u_{0 \eps},v_{0 \eps},w_{0 \eps}) \in X_s \times Y_{s-2} \times Y_{s-2}$ there exists a unique maximal mild solution $(u,v,w) = (u_{\eps},v_{\eps},w_{\eps}) \in C(I_{\eps}, X_s \times Y_{s-2} \times Y_{s-2})$ of~\eqref{eq: original eq, abstract} with $I_{\eps} \subset I$, where $I = [0,T_0]$ is the interval from the previous theorem.  
\end{theorem}

\begin{proof}
Clearly, the linear part $A = A_{\eps} = \operatorname{diag}(i\Delta,\, \eps^{-1}\Lambda)$ of the equation~\eqref{eq: original eq, abstract} given by %with $A(u,v,w) := (i\Delta u, \eps^{-1} w, -\eps^{-1} v)$
\begin{gather*}
%A(u,v,w) := (i\Delta u, \, \eps^{-1} \Lambda (v,w)) = (i\Delta u, \, \eps^{-1} w, \, -\eps^{-1} v)  \qquad ((u,v,w) \in X_{s+2} \times Y_{s-2} \times Y_{s-2})
A(u,v,w) := (i\Delta u, \, \eps^{-1} \Lambda (v,w)) = (i\Delta u, \, \eps^{-1} w, \, -\eps^{-1} v) \\ ((u,v,w) \in X_{s+2} \times Y_{s-2} \times Y_{s-2})
\end{gather*}
is the generator of a strongly continuous group $(e^{A_{\eps}t})_{t \in \R} = (\operatorname{diag}(e^{i\Delta t},\, e^{\eps^{-1}\Lambda t}))_{t \in \R}$
in $X_s \times Y_{s-2} \times Y_{s-2}$ by virtue of Lemma~\ref{lm: X_s basics}~(iii) and Lemma \ref{lm: group uniformly bd in eps}. 
Also, the nonlinear part $f = f_{\eps}$ of the equation~\eqref{eq: original eq, abstract} given by 
\begin{align*}
f(u,v,w) := (-iuv,0,\eps^{-1} \Delta^{-1}(|u|^2)) 
%\qquad ((u,v,w) \in X_s \times Y_{s-2} \times Y_{s-2})
\end{align*}
is a map from $X_s \times Y_{s-2} \times Y_{s-2}$ into itself and Lipschitz on bounded subsets by virtue of Lemma~\ref{lm: Y_s basics}~(ii)-(iii) and Lemma~\ref{lm: X_s basics}~(ii). 
So, for $I = [0,T_0]$ as in Theorem~\ref{thm: approx eq mild sol}, the standard existence and uniqueness result for mild solutions (Theorem~6.1.4 in~\cite{Pazy}) %or Theorem~8.6 in~\cite{ISEM}) 
implies that there is a unique maximal mild solution $(u_{\eps},v_{\eps},w_{\eps}): I_{\eps} \to X_s \times Y_{s-2} \times Y_{s-2}$ of~\eqref{eq: original eq, abstract} with $I_{\eps} \subset I$. In other words, the integral equation
\begin{align} \label{eq: def mild sol, original eq}
\begin{pmatrix}u\\v\\w\end{pmatrix}(t) 
%= e^{A_{\eps}t} \begin{pmatrix}u_0\\v_0\\w_0\end{pmatrix} + \int_0^t e^{A_{\eps}(t-r)} f_{\eps}(u(r),v(r),w(r)) \,dr
= e^{A_{\eps}t} \begin{pmatrix}u_0\\v_0\\w_0\end{pmatrix} + \int_0^t e^{A_{\eps}(t-r)} \begin{pmatrix}-i u(r)v(r)\\ 0\\ \eps^{-1} \Delta^{-1}(|u(r)|^2)\end{pmatrix} \,dr
\end{align}
has a unique maximal solution $(u,v,w) = (u_{\eps},v_{\eps},w_{\eps}) \in C(I_{\eps}, X_s \times Y_{s-2} \times Y_{s-2})$ with $I_{\eps} \subset I$.
\end{proof}

In the situation of the above theorem, we also obtain classical solvability by Theorem~6.1.5 in \cite{Pazy}:

\begin{corollary} \label{cor: original eq class sol}
%In the situation of the above theorem, 
If $s \in [2,\infty)$ and $\eps >0$ and if $(u_0,v_0,w_0)$ and $(u,v,w) = (u_{\eps},v_{\eps},w_{\eps})  \in C(I_{\eps}, X_s \times Y_{s-2} \times Y_{s-2})$ are as in the above theorem, then %$(u,v,w)  \in C^1(I_{\eps}, X_{s-2} \times Y_{s-2} \times Y_{s-2})$ 
$(u,v,w)$ belongs to $C^1(I_{\eps}, X_{s-2} \times Y_{s-2} \times Y_{s-2})$ and is a classical solution of~\eqref{eq: original eq, abstract}.
\end{corollary}

For the subsequent estimates, we additionally have to control the second-order time derivatives.
\begin{lemma} \label{lm: reg of v, V}
Suppose $U \in C(I,X_{s+2})$ is as in Theorem~\ref{thm: approx eq mild sol} with $s \in [2,\infty)$ and $V(t) := \Delta^{-1}(|U(t)|^2)$ for $t \in I$. Suppose further $(u,v,w) = (u_{\eps},v_{\eps},w_{\eps}) \in C(I_{\eps}, X_s \times Y_{s-2} \times Y_{s-2})$ is as in Theorem~\ref{thm: original eq mild sol}  with the same $s \in [2,\infty)$ as above. Then $V \in C^2(I,Y_{s-2})$ and $v \in C^2(I_{\eps},Y_{s-2})$.
\end{lemma}

\begin{proof}
With the help of Lemma~\ref{lm: X_s basics}~(ii) and (iii) and of Lemma~\ref{lm: Y_s basics}~(ii) and~(iii), it follows from Corollary~\ref{cor: approx eq class sol} and~\eqref{eq: approx eq, abstract} that $U' \in C^1(I,X_{s-2})$ and hence $U \in C^1(I,X_s) \cap C^2(I,X_{s-2})$. We easily conclude from this by Lemma~\ref{lm: X_s-2 X_s in X_s-2} that $|U|^2 \in C^2(I,X_{s-2})$ and therefore $V = \Delta^{-1}(|U|^2)$ belongs to $C^2(I,Y_{s-2})$ by Lemma~\ref{lm: Y_s basics}~(iii). 
That $v$ belongs to $C^2(I_{\eps},Y_{s-2})$ is an immediate consequence of Corollary~\ref{cor: original eq class sol} and~\eqref{eq: original eq, abstract}. 
\end{proof}

\section{Approximation error and approximation result}
\label{sec: approx error and approx result}

In this section, we are going to bound  the approximation error, that is the difference between the solutions $u = u_{\eps}$, $v = v_{\eps}$ of the original equations~\eqref{polaron1}-\eqref{polaron2} and the solutions $U$, $V$ of the approximate equations~\eqref{eq: hs1}-\eqref{eq: hs2}. We show that this difference -- measured in the right norm -- remains of order $\eps$ for all times $t \in [0,T_0]$ provided it was of order $\eps$ at the initial time $0$, thus establishing the desired approximation result.

\subsection{Integral equations and estimates for the error}

%As a first preparatory step, we derive ...
We first derive integral equations for the scaled approximation errors
\begin{gather*}
R_u(t) = R_{u,\eps}(t) := \eps^{-1} (u(t)-U(t)), \qquad R_v(t) = R_{v,\eps}(t) := \eps^{-1} (v(t)-V(t)), \\
R_w(t) = R_{w,\eps}(t) := \eps R_v'(t) = v'(t)-V'(t),
\end{gather*}
%\begin{eqnarray*}
%R_u(t) & = & R_{u,\eps}(t) := \eps^{-1} (u(t)-U(t)), \\
%R_v(t) & = & R_{v,\eps}(t) := \eps^{-1} (v(t)-V(t)), \\
%R_w(t) & = & R_{w,\eps}(t) := \eps R_v'(t) = v'(t)-V'(t),
%\end{eqnarray*}
where $U \in C(I,X_{s+2})$ and $(u,v,w) = (u_{\eps},v_{\eps},w_{\eps})  \in C(I_{\eps}, X_s \times Y_{s-2} \times Y_{s-2})$ are mild solutions of~\eqref{eq: approx eq, abstract} and~\eqref{eq: original eq, abstract} with $s \in [2,\infty)$ and where $V = \Delta^{-1}(|U|^2)$.
With the help of Corollary~\ref{cor: approx eq class sol} and~\ref{cor: original eq class sol} and Lemma~\ref{lm: reg of v, V}, we obtain 
\begin{align} \label{eq: int eq R_u}
R_u(t) = e^{i\Delta t}R_u(0) -i \int_0^t e^{i\Delta (t-r)} f_u(r)\,dr
\end{align}
for all $t \in I_{\eps}$, where 
%$f_u(r) = f_{u,\eps}(r) := R_u(r)V(r) + U(r)R_v(r) + \eps R_u(r)R_v(r)$, 
\begin{align*}
f_u(r) = f_{u,\eps}(r) := R_u(r)V(r) + U(r)R_v(r) + \eps R_u(r)R_v(r).
\end{align*}
and
\begin{align} \label{eq: int eq R_v,R_w}
\begin{pmatrix}R_v(t)\\R_w(t)\end{pmatrix} 
= e^{\eps^{-1} \Lambda t} \begin{pmatrix}R_v(0)\\R_w(0)\end{pmatrix}  +  \int_0^t e^{\eps^{-1} \Lambda (t-r)} \begin{pmatrix}0\\\eps^{-1}f_v(r)-V''(r)\end{pmatrix}  \,dr
\end{align}
for all $t \in I_{\eps}$, where 
%$f_v(r) = f_{v,\eps}(r) := \Delta^{-1}\big( R_u(r)\ol{U(r)} + \ol{R_u(r)}U(r) + \eps |R_u(r)|^2 \big)$.
\begin{align*}
f_v(r) = f_{v,\eps}(r) := \Delta^{-1}\big( R_u(r)\ol{U(r)} + \ol{R_u(r)}U(r) + \eps |R_u(r)|^2 \big).
\end{align*}

We now derive from the integral equations~\eqref{eq: int eq R_u} and~\eqref{eq: int eq R_v,R_w} integral inequalities which are implicit in the sense that the scaled approximation errors $R_u$ and $(R_v,R_w)$ -- measured in the norm of $X_s$ and $Y_{s-2} \times Y_{s-2}$ respectively -- %measured in suitable norms -- 
show up on both sides of the inequalities. In order to get rid of the dangerous $\eps^{-1}$ in front of $f_v$ in~\eqref{eq: int eq R_v,R_w} we perform an integration by parts. 
%
%Suppose $(U,V)$, $(u,v,w) = (u_{\eps},v_{\eps},w_{\eps})$, $R_u = R_{u,\eps}$ and $(R_v,R_w) = (R_{v,\eps},R_{w,\eps})$ are as in Proposition~\ref{prop: int eq}. 

\begin{proposition} \label{prop: int ineq}
Set
\begin{align*}
S_{u,\eps}(t) := \sup_{r \in [0,t]} \norm{R_u(r)}_{X_s}, 
\qquad S_{(v,w),\eps}(t) := \sup_{r \in [0,t]} \norm{(R_v(r),R_w(r))}_{Y_{s-2} \times Y_{s-2}},
\end{align*}
and $S_{\eps}(t) := S_{u,\eps}(t) + S_{(v,w),\eps}(t)$ for $t \in I_{\eps}$. Then there is a constant $C = C_s$ such that for all $\eps \in (0,\infty)$ and all $t \in I_{\eps}$
\begin{eqnarray*}
S_{u,\eps}(t) & \le & C \Big( S_{\eps}(0) + \int_0^t S_{\eps}(r) + \eps S_{\eps}(r)^2 \,dr \Big),
\\
S_{(v,w),\eps}(t) & \le & C \Big( S_{\eps}(0) + 1 + S_{u,\eps}(t) + \eps S_{u,\eps}(t)^2   
+ \int_0^t S_{\eps}(r) + \eps S_{\eps}(r)^2 + \eps^2 S_{\eps}(r)^3 \, dr \Big).
\end{eqnarray*}
\end{proposition}
%Then there is a constant $C = C_s$ such that
%\begin{align*}
%S_{u,\eps}(t) \le C \Big( S_{\eps}(0) + \int_0^t S_{\eps}(r) + \eps S_{\eps}(r)^2 \,dr \Big)  % \qquad (t \in I_{\eps})
%\end{align*}
%for all $t \in I_{\eps}$ and $\eps > 0$, and such that
%\begin{align*}
%S_{(v,w),\eps}(t) \le C \Big( S_{\eps}(0) + 1 + S_{u,\eps}(t) + \eps S_{u,\eps}(t)^2  
%+ \int_0^t S_{\eps}(r) + \eps S_{\eps}(r)^2 + \eps^2 S_{\eps}(r)^3 \, dr \Big)
%\end{align*}
%for all $t \in I_{\eps}$ and $\eps > 0$.
%
%
%
\begin{proof}
% 1. Absch\"atzung von S_{u,\eps}
It follows from Lemma~\ref{lm: Y_s basics}~(ii) that $f_u$  belongs to $C(I_{\eps},X_s)$ and satisfies the estimate
\begin{gather} 
\norm{f_u(r)}_{X_s} \le C \Big( \norm{R_u(r)}_{X_s} \norm{V(r)}_{Y_{s-2}} + \norm{U(r)}_{X_s} \norm{R_v(r)}_{Y_{s-2}} \notag \\
+ \eps \norm{R_u(r)}_{X_s} \norm{R_v(r)}_{Y_{s-2}} \Big)
\label{eq: estimate f_u}
\end{gather}
for all $r \in I_{\eps}$. 
%
%In view of~\eqref{eq: int eq R_u}, the asserted estimate for $S_{u,\eps}$ now follows by~\eqref{eq: estimate f_u} and Lemma~\ref{lm: X_s basics}~(iii) %unitarity of $e^{-i\Delta t}$
%and by $\sup_{r \in I} \norm{U(r)}_{X_s} < \infty$ and $\sup_{r \in I} \norm{V(r)}_{Y_{s-2}} < \infty$. 
Since $\sup_{r \in I} \norm{U(r)}_{X_s} < \infty$ and $\sup_{r \in I} \norm{V(r)}_{Y_{s-2}} < \infty$, the asserted estimate for $S_{u,\eps}$ now follows from~\eqref{eq: int eq R_u} with the help of~\eqref{eq: estimate f_u} and Lemma~\ref{lm: X_s basics}~(iii). %unitarity of $e^{-i\Delta t}$
%
% 2. Absch\"atzung von S_{(v,w),\eps}
Since $R_u \in C(I_{\eps},X_s) \cap C^1(I_{\eps},X_{s-2})$ and $U \in C^1(I,X_s)$ by Corollary~\ref{cor: approx eq class sol} and~\ref{cor: original eq class sol}, it follows from Lemma~\ref{lm: Y_s basics}~(iii) and Lemma~\ref{lm: X_s-2 X_s in X_s-2} that $f_v$  belongs to $C^1(I_{\eps},Y_{s-2})$. We can therefore integrate by parts in~\eqref{eq: int eq R_v,R_w} and thus obtain
\begin{align} \label{eq: int eq R_v,R_w, after part int}
\begin{pmatrix}R_v(t)\\R_w(t)\end{pmatrix} 
&= e^{\eps^{-1} \Lambda t} \begin{pmatrix}R_v(0)\\R_w(0)\end{pmatrix} - \int_0^t e^{\eps^{-1} \Lambda (t-r)} \begin{pmatrix}0\\V''(r)\end{pmatrix}  \,dr  \notag \\
&\quad -e^{\eps^{-1} \Lambda (t-r)} \Lambda^{-1} \begin{pmatrix}0\\f_v(r)\end{pmatrix} \Big|_{r=0}^{r=t} + \int_0^t e^{\eps^{-1} \Lambda (t-r)} \Lambda^{-1} \begin{pmatrix}0\\f_v'(r)\end{pmatrix}  \,dr
\end{align}
for all $t \in I_{\eps}$, where 
\begin{align*}
f_v'(r) &= \Delta^{-1}\Big( R_u(r)\ol{U'(r)} + R_u'(r) \big(\ol{U(r)} + \eps \ol{R_u(r)} \big) \Big) + \text{c.c.} \notag \\
&= \Delta^{-1}\Big( R_u(r)\ol{U'(r)}\Big) + i\Delta^{-1}\Big( \big(\Delta R_u(r) - f_u(r)\big) \big(\ol{U(r)} + \eps \ol{R_u(r)}\big) \Big) + \text{c.c.}.%&= \Delta^{-1}\big( R_u(r)\ol{U'(r)}) -i\Delta^{-1}\big( (\Delta R_u(r)) (\ol{U(r)} + \eps \ol{R_u(r)}) \big) - i\Delta^{-1}\big( f_u(r) (\ol{U(r)} + \eps \ol{R_u(r)}) \big) + \text{c.c.}
\end{align*}
In the above equation, the symbol c.c.~stands for the complex conjugate of the terms on the left of it.  
With the help of Lemma~\ref{lm: Y_s basics}~(iii) and Lemma~\ref{lm: X_s basics}~(ii), we can estimate 
\begin{align} 
\norm{f_v(r)}_{Y_{s-2}} &\le C \Big( \norm{ R_u(r)\ol{U(r)} }_{X_s \cap L^1} + \eps \norm{R_u(r) \ol{R_u(r)} }_{X_s \cap L^1} \Big) \notag \\
&\le C \Big( \norm{R_u(r)}_{X_s} \norm{U(r)}_{X_s} + \eps \norm{R_u(r)}_{X_s}^2 \Big)
\label{eq: estimate f_v}
\end{align}
for all $r \in I_{\eps}$, 
and with the help of Lemma~\ref{lm: Y_s basics}~(iii) and Lemmas~\ref{lm: X_s-2 X_s in X_s-2} and~\ref{lm: X_s basics}~(ii)-(iii), we can estimate
\begin{gather} 
\norm{f_v'(r)}_{Y_{s-2}} \le C \Big( \norm{ R_u(r)\ol{U'(r)} }_{X_s \cap L^1} + \norm{ (\Delta R_u(r)) (\ol{U(r)} + \eps \ol{R_u(r)}) }_{X_{s-2} \cap L^1} \notag \\ + \norm{ f_u(r) (\ol{U(r)} + \eps \ol{R_u(r)}) }_{X_s \cap L^1} \Big) \notag \\
\le C \Big( \norm{R_u(r)}_{X_s} \norm{U'(r)}_{X_s} + \norm{R_u(r)}_{X_s} \big( \norm{U(r)}_{X_s} + \eps \norm{R_u(r)}_{X_s} \big) \notag \\ 
+ \norm{f_u(r)}_{X_s} \big( \norm{U(r)}_{X_s} + \eps \norm{R_u(r)}_{X_s} \big) \Big) 
\label{eq: estimate f_v'}
%C \Big( \norm{R_u(r)}_{X_s} \norm{U'(r)}_{X_s} + \big( \norm{R_u(r)}_{X_s} + \norm{f_u(r)}_{X_s} \big) \big( \norm{U(r)}_{X_s} + \eps \norm{R_u(r)}_{X_s} \big) \Big)
\end{gather}
for all $r \in I_{\eps}$. 
%
%In view of~\eqref{eq: int eq R_v,R_w, after part int}, the asserted estimate for $S_{(v,w),\eps}$ now follows by~\eqref{eq: estimate f_v}, \eqref{eq: estimate f_v'}, \eqref{eq: estimate f_u} and Lemma~\ref{lm: group uniformly bd in eps} and by $\sup_{r \in I} \norm{U(r)}_{X_s} + \norm{U'(r)}_{X_s} < \infty$ and $\sup_{r \in I} \norm{V(r)}_{Y_{s-2}} + \norm{V''(r)}_{Y_{s-2}} < \infty$.  
Since $\sup_{r \in I} \norm{U(r)}_{X_s} + \norm{U'(r)}_{X_s} < \infty$ and $\sup_{r \in I} \norm{V(r)}_{Y_{s-2}} + \norm{V''(r)}_{Y_{s-2}} < \infty$, the asserted estimate for $S_{(v,w),\eps}$ now follows from~\eqref{eq: int eq R_v,R_w, after part int} with the help of~\eqref{eq: estimate f_v}, \eqref{eq: estimate f_v'}, \eqref{eq: estimate f_u} and Lemma~\ref{lm: group uniformly bd in eps}. 
\end{proof}

%\subsection{Statement and proof of the approximation theorem}
\subsection{Approximation result} \label{sec: approx result}

With the help of Gronwall's lemma, we finally turn the implicit estimates for the approximation error just established into explicit estimates and thus obtain our approximation theorem. 
%With the preparations from the previous sections at hand, we can now finally prove our approximation theorem. %the approximation theorem we aim at.  %our main result. 
Choosing $s = 2$, we obtain the version of the theorem stated in the introduction. %Section~1. 

\begin{theorem} \label{thm: approx result}
Suppose $(U,V)$ and $(u,v,w) = (u_{\eps},v_{\eps},w_{\eps})$ are as in Lemma~\ref{lm: reg of v, V} and suppose further that the initial values satisfy
\begin{align}
\norm{u_{\eps}(0) - U(0)}_{X_s} + \norm{v_{\eps}(0)-V(0)}_{Y_{s-2}} + \eps \, \norm{v_{\eps}'(0)-V'(0)}_{Y_{s-2}} \le C_1 \eps
\end{align}
for all $\eps \in (0,\eps_0]$ with some $\eps_0 > 0$ and some constant $C_1 = C_{1, s}$. Then there is an $\eps_0' \in (0,\eps_0]$ and a constant $C_2 = C_{2,s}$ such that $I_{\eps} = I$ for all $\eps \in (0,\eps_0']$ and such that
\begin{align} \label{eq: main thm, estimate for all t}
\norm{u_{\eps}(t) - U(t)}_{X_s} + \norm{v_{\eps}(t)-V(t)}_{Y_{s-2}} + \eps \, \norm{v_{\eps}'(t)-V'(t)}_{Y_{s-2}} \le C_2 \eps
\end{align}
for all $t \in I$ and all $\eps \in (0,\eps_0']$.
\end{theorem}

\begin{proof}
We plug in the estimate for $S_{u,\eps}$ into the estimate for $S_{(v,w),\eps}$ from Proposition~\ref{prop: int ineq} and, by adding the resulting inequality to the inequality for $S_{u,\eps}$, %from Proposition~\ref{prop: int ineq}, 
we obtain the following inequality for $S_{\eps}$:
\begin{align*}
S_{\eps}(t) \le C \Big( S_{\eps}(0) + 1 + \int_0^t S_{\eps}(r) + \eps S_{\eps}(r)^2 + \eps^2 S_{\eps}(r)^3 + \eps^3 S_{\eps}(r)^4 \, dr \Big)
\end{align*}
for all $t \in I_{\eps}$ and all $\eps \in (0,\infty)$. Since $S_{\eps}(0) \le C_1$ for all $\eps \in (0, \eps_0]$ by assumption, we therefore have that %conclude that %therefore have %thus obtain
\begin{align}
S_{\eps}(t) \le C + C \int_0^t p(\eps S_{\eps}(r)) S_{\eps}(r) \,dr
\end{align}
for all $t \in I_{\eps}$ and all $\eps \in (0, \eps_0]$, where %$p$ denotes the polynomial with 
$p(\xi) := 1 + \xi + \xi^2 + \xi^3$ for $\xi \in \R^3$. So, by Gronwall's lemma, we obtain
\begin{align} \label{eq: main thm, 1}
S_{\eps}(t) \le C e^{C \int_0^t p(\eps S_{\eps}(r)) \,dr}
\end{align}
for all $t \in I_{\eps}$ and all $\eps \in (0, \eps_0]$. Set $M := C e^{T_0} + 2$ with $C$ being the  constant in~\eqref{eq: main thm, 1} and choose $\eps_0' \in (0,\eps_0]$ such that
\begin{align} \label{eq: main thm, 2}
C e^{p(\eps_0' M) T_0} \le C e^{T_0} + 1 = M-1.
\end{align}
%(which is possible because $p(\eps M) T_0 \to T_0$ as $\eps \searrow 0$). 
Also, for $\eps \in (0,\eps_0']$ set
\begin{align*}
b_{\eps} := \sup \big\{ t \in I_{\eps}: S_{\eps}(r) \le M \text{ for all } r \in [0,t] \big\}.
\end{align*}
We then have, for all $t \in [0,b_{\eps})$ and all $\eps \in (0,\eps_0']$, that
\begin{align} \label{eq: main thm, 3}
S_{\eps}(t) \le C e^{C \int_0^t p(\eps S_{\eps}(r)) \,dr}
\le C e^{p(\eps_0' M) T_0} \le M-1
\end{align}
by virtue of~\eqref{eq: main thm, 1} and~\eqref{eq: main thm, 2}. If now $b_{\eps}$ was strictly less than $\sup I_{\eps}$ for some $\eps \in (0,\eps_0']$, then from~\eqref{eq: main thm, 3} we would obtain, using the continuity of $S_{\eps}$, a contradiction to the definition of $b_{\eps}$. So, $b_{\eps} = \sup I_{\eps}$ for all $\eps \in (0,\eps_0']$.
%\begin{align*}
%b_{\eps} = \sup I_{\eps}.
%\end{align*} 
It follows from this and from~\eqref{eq: main thm, 3} that 
\begin{align} \label{eq: main thm, 4}
\sup_{t \in I_{\eps}} S_{\eps}(t) = \sup_{t \in [0,b_{\eps})} S_{\eps}(t) \le M-1
\end{align}
and hence
\begin{align} \label{eq: main thm, 5}
\sup_{t \in I_{\eps}} \norm{(u_{\eps}(t),v_{\eps}(t),w_{\eps}(t))}_{X_s \times Y_{s-2} \times Y_{s-2}} \le C + \eps_0' \sup_{t \in I_{\eps}} S_{\eps}(t) < \infty
\end{align}
for all $\eps \in (0,\eps_0']$. Since $(u_{\eps},v_{\eps},w_{\eps})$ by definition is the maximal mild solution of~\eqref{eq: original eq, abstract} with $I_{\eps} \subset I$, it follows from~\eqref{eq: main thm, 5} by the standard blow-up result for mild solutions (Theorem~6.1.4 in~\cite{Pazy}) %or Theorem~8.6 in~\cite{ISEM}) 
that $I_{\eps}$ must be equal to $I$ for all $\eps \in (0,\eps_0']$. So, invoking~\eqref{eq: main thm, 4} again, we see that 
\begin{align}
\sup_{t \in I} S_{\eps}(t) = \sup_{t \in I_{\eps}} S_{\eps}(t) \le M-1 =: C_2
\end{align}
for all $\eps \in (0,\eps_0']$, and this immediately implies%yields
~\eqref{eq: main thm, estimate for all t}.
\end{proof}

\section{Concluding remarks} \label{remnormalform}

We close this paper with some remarks on  
the connection of the presented approach to normal form transformations. Instead of the approach pursued above, one can try to get rid of the dangerous term
$ \varepsilon^{-1}\Delta^{-1}  (U
\overline{R_u} + \overline{U}R_u)  $ in the equation~\eqref{eq: int eq R_v,R_w} for $R_{(v,w)} := (R_v,R_w)$ 
 by a near-identity change of coordinates of the form
\begin{align*}
\widetilde{R}_{(v,w)} = {R}_{(v,w)} + B(U,R_u)
\end{align*}
where $ B $ is a symmetric bilinear mapping.
Inserting this transformation into the $ R_{(v,w)} $ equation 
yields
\begin{align*}
\partial_t \widetilde{R}_{(v,w)} & =  \varepsilon^{-1}  \Lambda \widetilde{R}_{(v,w)}  - \varepsilon^{-1} \Lambda B(U,R_u) + B(-i\Delta U,R_u) 
+ B( U,-i\Delta R_u) \\ 
&\quad   +  \varepsilon^{-1}\Delta^{-1} \left( U
\overline{R_u} + \overline{U}R_u\right)  + \text{h.o.t.},
\end{align*}
where h.o.t.~stands for the higher-order terms. So, we have to find a bilinear mapping $ B $ 
such that  
$$ 
- \varepsilon^{-1} \Lambda B(U,R_u) + B(-i\Delta U,R_u) 
+ B( U,-i\Delta R_u)  +  \varepsilon^{-1}\Delta^{-1} \left( U
\overline{R_u} + \overline{U}R_u\right) = 0,
$$ 
which is not possible, however, since the non-resonance condition 
$$ \inf_{k,l} |\pm \varepsilon^{-1} + (k-l)^2 -l^2| > 0  $$ 
is not satisfied.
The approach we took above corresponds to the  choice $ B(U,R_u) =  \Lambda^{-1} \Delta^{-1} \left( U
\overline{R_u} + \overline{U}R_u\right) = \mathcal{O}(1)$.
It allowed us to eliminate the terms of order 
$ \mathcal{O}(\varepsilon^{-1}) $
such that after the transform we had
\begin{align*}
\partial_t \widetilde{R}_{(v,w)} & = \varepsilon^{-1}  \Lambda \widetilde{R}_{(v,w)}   + B(-i\Delta U,R_u) 
+ B( U,-i\Delta R_u)  + \dotsb  \\ & =  \varepsilon^{-1}  \Lambda \widetilde{R}_{(v,w)} 
+ \mathcal{O}(1).
\end{align*}
We finally remark that the energy approach chosen in \cite{DSS15} 
for a similar limit in the Klein--Gordon--Zakharov system 
can only be used for \eqref{polaron1}-\eqref{polaron2}
in space dimensions $ d \geq 5 $ due to the occurrence of $ \Delta^{-1}$ which maps $ L^1(\mathbb{R}^d)  \cap L^2(\mathbb{R}^d)  $ into $ L^2(\mathbb{R}^d)  $ only for space dimensions $ d \geq 5 $.

%\bibliography{literature}
%\bibliographystyle{alpha}

\end{document}